\documentclass[journal]{IEEEtran}

\usepackage{booktabs}
\usepackage{siunitx}
\usepackage{graphicx}
\usepackage{cite}
\usepackage{subcaption}
\usepackage{amsmath,amssymb,amsfonts}
\usepackage{fancyhdr}
\usepackage{algorithmic}
\usepackage{authblk}
\usepackage[ruled,linesnumbered,lined,boxed,commentsnumbered]{algorithm2e}
\usepackage{graphicx}
\usepackage{hyperref}
\usepackage{textcomp}
\usepackage{enumitem}
\usepackage{pifont}
\usepackage{commath}
\usepackage{xcolor}
\usepackage{upgreek}
\usepackage{multirow}
\usepackage{commath}
\usepackage{makecell}
\usepackage{threeparttable}
\usepackage{bm}
\newtheorem{definition}{Definition}
\newtheorem{theorem}{Theorem}
\newtheorem{proof}{Proof}

\usepackage{booktabs}

\usepackage{tikz}
\newcommand{\solidcircle}{\tikz\filldraw[black] (0,0) circle (2.5pt);}

\newcommand{\hollowcircle}{\tikz\draw[black] (0,0) circle (2.5pt);}

\newcommand{\halfcircle}{%
  \tikz[baseline=-0.5ex]{%
    \fill (0,0) -- (90:2.5pt) arc (90:270:2.5pt) -- cycle; 
    \draw (0,0) circle (2.5pt); 
  }%
}

\SetCommentSty{mycommfont}

\def\BibTeX{{\rm B\kern-.05em{\sc i\kern-.025em b}\kern-.08em
    T\kern-.1667em\lower.7ex\hbox{E}\kern-.125emX}}

\begin{document}
%
\title{Pacifo: An Efficient Fully Homomorphic Encryption-based Lightweight ResNet for Trustworthy Machine Learning as a Service}

\author{Yang Li, Gang He, Ruichen Zhang, Guangyuan Liu, and Yinqiu Liu
\thanks{Y. Li and G. Liu are with the College of Computing and Data Science, the Energy Research Institute @ NTU, Interdisciplinary Graduate Program, Nanyang Technological University, Singapore (e-mail: yang048@e.ntu.edu.sg; liug0022@e.ntu.edu.sg).
} 
\thanks{G. He is with the School of Communication and Information Engineering, Shanghai University,
Shanghai, China (e-mail: hegang@shu.edu.cn).
} 
\thanks{R. Zhang, and Y. Liu are with the College
of Computing and Data Science, Nanyang Technological University,
Singapore (e-mails: ruichen.zhang@ntu.edu.sg; yinqiu001@e.ntu.edu.sg)}
}
\maketitle

\begin{abstract}
Machine Learning as a Service (MLaaS) has become a popular paradigm for deploying AI applications. However, it raises significant privacy concerns that the raw client data needs to be uploaded to the server for inference services. 
Private inference using Fully Homomorphic Encryption (FHE) offers a solution by allowing computations directly on encrypted data, with only the final results decrypted on a client side. However, current practical implementations still suffer from large computational overhead and high latency due to the cost of FHE operations.
In this paper, we propose Pacifo, a framework that adapts customized fine-tuning strategies to efficiently evaluate the ResNet20 model in the context of FHE. Specifically, Pacifo introduces a public-private layer partitioning scheme and an efficient ReLU approximation, significantly reducing the number of private layers and costly FHE operations during inference. It also offers flexible security configurations to accommodate varying client privacy requirements. Compared to state-of-the-art FHE implementations of ResNet20, Pacifo achieves a speedup of $3.7-21.7\times$ in runtime while effectively preserving $95\%-99\%$ of the original plaintext model accuracy across downstream tasks. Pacifo is open-sourced at \url{https://github.com/BryceYangLi/Pacifo}.

\end{abstract}

\begin{IEEEkeywords}
Privacy-preserving machine learning, fully homomorphic encryption, private inference, ResNet20, fine-tuning
\end{IEEEkeywords}

%
\IEEEpeerreviewmaketitle

\section{Introduction}

In recent years, deep learning-based models have shown remarkable capabilities across a wide range of tasks, including image recognition~\cite{devi2021review,hasan2024image}, sentiment analysis~\cite{alslaity2024machine}, and content generation~\cite{lu2023machine}. In this case, \textit{Machine Learning as a Service} (MLaaS)~\cite{ribeiro2015mlaas,philipp2020machine} has emerged as a popular business use case to facilitate the development of diverse AI applications. Fig.~\ref{fig:MLaaS} presents a typical MLaaS workflow: the well-trained model is often deployed on the remote server, and the clients need to upload their data to the server and receive the inference result in return. This setup allows clients to avoid the cost of training computationally intensive models locally, while the server retains ownership of its valuable model.

However, the data handling process in MLaaS could raise significant privacy concerns. Specifically, MLaaS often operates through an online platform or APIs, requiring clients to transmit input data to the server.
This reliance could lead to potential misuse of private data, such as unauthorized processing or resale to third parties. Even when the server is considered trustworthy, clients with sensitive data (e.g., medical records or financial information) may be unwilling to grant direct access to their data. These concerns have triggered responses scuh as Italy’s temporary ban of ChatGPT~\cite{mauran2023whoops,lomas2023italy}. Hence, different private outsourced computation techniques have been applied to MLaaS to mitigate data privacy concerns, e.g., Differential Privacy (DP)~\cite{dwork2014algorithmic}, secure Multi-Party Computation (MPC)~\cite{yao1982protocols}, and Homomorphic Encryption (HE)~\cite{gentry2009fully}. These approaches belong to the broader field of Privacy-Preserving Machine Learning (PPML) and focus on \textit{private inference} (PI), which means performing model inference service without exposing client data~\cite{li2024survey}.

\textbf{FHE for Private Inference.}
Among existing solutions, FHE-based approaches are considered highly promising. In this setting, clients encrypt their private inputs before transmission, and only decrypt the final inference results returned by the server, ensuring that the server gains no knowledge of the input~\cite{podschwadt2022survey}. The privacy guarantees of FHE are grounded in the hardness of well-established cryptographic assumptions, i.e., the Ring Learning With Errors (RLWE) problem~\cite{lyubashevsky2010ideal}.
In addition, FHE-based protocols are usually \textit{non-interactive}, which means no additional communication is required between the client and server during the inference process. Thereafter, there are several studies
on HE-based private inference~\cite{gilad2016cryptonets,chou2018faster,jiang2018secure,dathathri2019chet,boemer2019ngraph,lou2021hemet,park2022aespa,rovida2024transformer,zhang2024secure,moon2024thor}, with recent works~\cite{lee2022privacy,lee2022low,benamira2023tt,kim2023optimized,rovida2024encrypted} targeting deep and wide Convolutional Neural Networks (CNNs). Specifically, CNNs excel at capturing two-dimensional spatial features, making them effective for a wide range of image processing tasks and a popular choice for private inference~\cite{he2025epcnn}.



\begin{figure}[t]
    \centering
\includegraphics[width=0.85\linewidth]{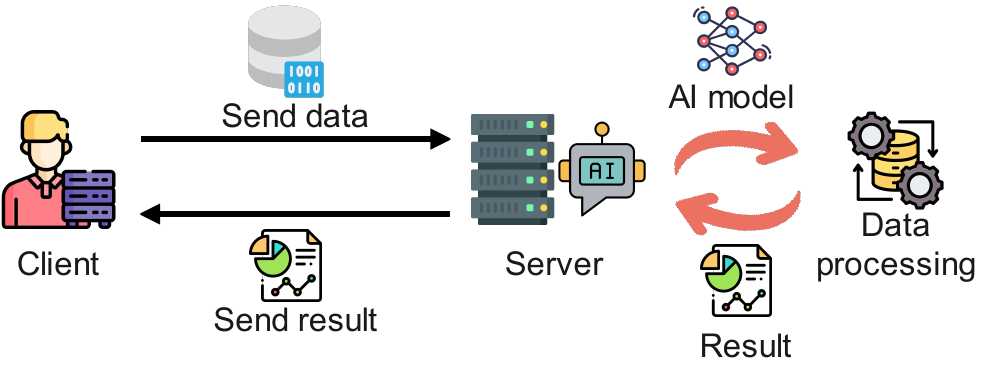}
    \caption{MLaaS without privacy protection}
    \label{fig:MLaaS}
\end{figure}

\textbf{Challenges in FHE-Based CNN Inference.} Most existing HE-based solutions~\cite{lee2022privacy,lee2022low,benamira2023tt,kim2023optimized,rovida2024encrypted} utilize the Cheon-Kim-Kim-Song (CKKS) scheme~\cite{cheon2017homomorphic} 
(except Benamira~\emph{et~al.}\cite{benamira2023tt} 
with TFHE~\cite{chillotti2016faster}). CKKS supports approximate arithmetic over real and complex numbers, making it well-suited for PPML. Nonetheless, current implementations on CNNs still face the following challenges:
\begin{enumerate}
    \item Advanced CNNs such as ResNet involve many convolutions and non-linear activations, which require a large number of expensive homomorphic computations and customized polynomial approximations in the context of FHE. This results in significant computational overhead and high runtime costs. For instance, Kim~\emph{et~al.}~\cite{kim2023optimized} reported a runtime of $368$ seconds to evaluate a single CIFAR-10 image using ResNet20.

    \item CNNs often have deep architectures—for example, even the lightweight ResNet20 consists of 20 network layers. Designing an FHE circuit to evaluate the entire model typically requires frequent and costly \textit{bootstrapping} operations to manage noise during computation. In Lee~\emph{etal.}\cite{lee2022low} and Kim~\emph{etal.}\cite{kim2023optimized}, bootstrapping accounts for 72.6\% and 82.3\% of the total runtime, respectively.
    \item Existing solutions~\cite{lee2022privacy,lee2022low,benamira2023tt,kim2023optimized,rovida2024encrypted} typically adopt a fixed security standard in HE (e.g., $128$-bit), which is reasonable but lacks flexibility. We believe that the security levels should align with different application scenarios. Highly sensitive data, e.g., government communications, requires stronger security (e.g., $256$-bit) to counter potential quantum threats.
\end{enumerate}
To address challenges 1) and 2), prior studies have targeted computationally expensive operations, e.g., using single instruction multiple data (SIMD) to accelerate the linear computations~\cite{lee2022privacy,lee2022low,kim2023optimized} and proposing HE-friendly approximations for non-linear operations~\cite{rovida2024encrypted}. 
They have also optimized the design and placement of bootstrapping to further reduce overhead~\cite{lee2021high}.
However, despite significant advancements in efficiency, these protocols remain impractical for real-world applications due to high runtime costs.  Besides, none of them considers challenge 3) for further optimizing security standards. Given these limitations, we are motivated to re-consider how to design an efficient and practical private inference framework from a higher level, while ensuring the model performance.


\textbf{Our Solutions.} 
In this paper, we present \textit{Pacifo}, a client-server private inference framework that minimizes FHE overhead through high-level circuit design. Pacifo leverages open-source pretrained models and strategically modifies the fine-tuning process to shift most of the computations outside the FHE domain without compromising security. As a result, the associated computational cost becomes relatively negligible. Furthermore, it introduces aggressive ReLU approximations and incorporates knowledge distillation (KD) during fine-tuning to preserve accuracy. Additionally, Pacifo supports flexible security levels, allowing clients to balance efficiency and security standards based on their specific requirements.

\subsection{Contributions}
We customize Pacifo for the popular ResNet20 model (See Section~\ref{Subsec:resnet20}), which is widely used for image classification tasks due to its lightweight architecture and effective use of residual connections that mitigate vanishing gradients.
Our main contributions can be summarized as follows:
\begin{enumerate}
    \item \textbf{Layer Partitioning for Efficient Private Inference.} 
    We propose a layer-freezing-based fine-tuning strategy that partitions ResNet20 into public and private layers. Clients process inputs through public layers to generate intermediate features, which are encrypted and sent to the server. The server computes on the encrypted features using private layers and returns the encrypted output for the client to decrypt. Adjusting the partition offers a flexible trade-off between performance and efficiency.
    \item \textbf{HE-friendly Approximation for ReLU Activations.} 
    We present a highly efficient approximation for the expensive non-linear ReLU function. Unlike prior work, we trade partial functional accuracy for significantly reduced computational overhead, consuming only two multiplicative depths in HE. By leveraging knowledge distillation, we experimentally show that this simplification results in only minimal performance degradation.
    \item \textbf{Customized Security Standards.} 
    We have also investigated the relationship between security and HE configurations using the Learning With Errors (LWE)-estimator, enabling clients to choose different security levels based on privacy requirements. For instance, clients with higher privacy needs can opt for stronger security (e.g., 192-bit security) to effectively defend against potential quantum attacks. 
    \item \textbf{An End-to-End Implementation using FHE.} 
    We also provide a full-fledged implementation of ResNet20 using the \textit{residue number system} (RNS) variant of CKKS, i.e., RNS-CKKS, and evaluate its performance, including runtime, communication overhead, and model accuracy. Our approach could achieve a runtime of 16.9 seconds and 10.5 MB of communication—over 20× faster than previous studies~\cite{kim2023optimized,rovida2024encrypted}, while maintaining accuracy comparable to the plaintext model. Our source codes are also available at \url{https://github.com/BryceYangLi/Pacifo}.
\end{enumerate}

\subsection{Organization} 
We summarize the related work and present the preliminaries in Section~\ref{Sec:related_work} and Section~\ref{Sec:preliminary}, respectively. The basic design for an HE-based Resnet20 circuit is presented in Section~\ref{Sec:basic_design}. Our fine-tuning and layer partition strategy is given in Section~\ref{Sec:layer_freezing}. Section~\ref{Sec:security_standards} discusses the customization of different security standards. Section~\ref{Sec:security_anlaysis} analyzes the security of our proposed framework. Finally, we demonstrate the experimental results in Section~\ref{Sec:experiments} and conclude the paper in Section~\ref{Sec:conclusion}.

\section{Related Work}\label{Sec:related_work}

This section reviews prior work on private inference targeting on CNNs. In general, Table~\ref{tab:related_work} and Fig.~\ref{fig:comparison} provide a summarization of the most related works compared to our proposed method. After that, we discuss these works in detail according to their underlying techniques.

\begin{table*}[ht]
\centering
\begin{threeparttable}
  \caption{Comparison of privacy-preserving schemes for private neural network inference.}
  \label{tab:related_work}
  \centering
  \begin{tabular}{l|cccccccc}
    \hline
    \textbf{Framework} & \makecell{\textbf{Cryptographic}\\\textbf{Method}} & \makecell{\textbf{Threat}\\\textbf{Model}} & \makecell{\textbf{Data}\\\textbf{Privacy}} & \makecell{\textbf{Model}\\\textbf{Privacy}} & \makecell{\textbf{Output}\\\textbf{Privacy}} & \makecell{\textbf{Non-}\\\textbf{Interactivity}}& \makecell{\textbf{Computation}\\\textbf{Efficiency}} & \makecell{\textbf{Communication}\\\textbf{Efficiency}} \\\hline
    CryptoNets~\cite{gilad2016cryptonets} & HE & Semi-honest & $\solidcircle$ &  $\hollowcircle$ & $\solidcircle$ & $\solidcircle$ & Low & Low \\
    Gazelle~\cite{juvekar2018gazelle} & HE,GC,SS & Semi-honest & $\solidcircle$ &  $\hollowcircle$ & $\solidcircle$ & $\hollowcircle$ & High & Low \\
    MUSE~\cite{lehmkuhl2021muse} & HE,GC,SS & Semi-honest & $\solidcircle$ &  $\solidcircle$ & $\solidcircle$ & $\hollowcircle$ & Medium & Low \\
    Lee~\emph{et~al.}~\cite{lee2022privacy} & HE & Semi-honest & $\solidcircle$ &  $\solidcircle$ & $\solidcircle$ & $\solidcircle$ & Low & High \\
    Lee~\emph{et~al.}~\cite{lee2022low} & HE & Semi-honest & $\solidcircle$ &  $\solidcircle$ & $\solidcircle$ & $\solidcircle$ & Low & High \\
    Sonic~\cite{liu2022securely} & SS & Honest-majority & $\halfcircle$ &  $\halfcircle$ &  $\halfcircle$ & $\hollowcircle$ & Medium & High \\
    Kim~\emph{et~al.}~\cite{kim2023optimized} & HE & Semi-honest & $\solidcircle$ &  $\solidcircle$ &  $\solidcircle$ & $\solidcircle$ & Medium & High \\
    Pio~\cite{yang2023efficient} & HE,GC,SS & Honest-majority &  $\halfcircle$ &  $\solidcircle$ &  $\halfcircle$ & $\hollowcircle$ & High & Medium \\
    Zhao~\emph{et~al.}~\cite{zhao2024pratical} & SS & Honest-majority &  $\halfcircle$ &  $\halfcircle$ &  $\halfcircle$ & $\hollowcircle$ & Medium & Medium \\
    Rovida~\emph{et~al.}~\cite{rovida2024encrypted} & HE & Semi-honest & $\solidcircle$ & $\solidcircle$ & $\solidcircle$ & $\solidcircle$ & Medium & High \\
    EPCNN~\cite{he2025epcnn} & HE & Semi-honest & $\solidcircle$ & $\solidcircle$ & $\solidcircle$ & $\hollowcircle$ & Medium & Low \\
    Pacifo (proposed) & HE & Semi-honest & $\solidcircle$ & $\solidcircle$ & $\solidcircle$ & $\solidcircle$ & High & High \\
    \bottomrule
  \end{tabular}
  \begin{tablenotes}[flushleft]
    \footnotesize
    \item[1] HE: Homomorphic Encryption; GC: Garbled Circuits; SS: Secret Sharing.
    \item[2] Symbols: $\solidcircle$ conditions met; $\halfcircle$ partially met; $\hollowcircle$ not met.
    \item[3] Efficiency levels are reported relative to our proposed Pacifo method.
  \end{tablenotes}
  \end{threeparttable}
\end{table*}

\subsection{HE-only Private Inference}

We first present private inference studies based on HE schemes only, as they are closely related to our work. HE-only methods typically operates in a \emph{non-interactive} manner: the client encrypts inputs, the server evaluates the network under encryption, and the client decrypts the result.

Early HE-only approaches such as CryptoNets~\cite{gilad2016cryptonets} and HCNN~\cite{al2020towards} demonstrated the feasibility of evaluating neural networks over encrypted data, but they lacked support for privacy-preserving nonlinear activations. Consequently, they either delegated nonlinear computations to the client or employed HE-friendly polynomial approximations, which compromised model privacy or reduced accuracy. Lee~\emph{et~al.}~\cite{lee2022privacy} employed more advanced RNS-CKKS~\cite{cheon2018full} and polynomial approximation methods to implement a ResNet20 model, achieving high accuracy comparable to the plaintext model at the cost of considerable computational overhead. They further optimized their bootstrapping design and HE parameter configurations in a subsequent work~\cite{lee2022low}, but it still requires half an hour of computations for a single inference. To improve the computational efficiency, Kim~\emph{et~al.}~\cite{kim2023optimized} optimized the encoding process in HE and enabled more efficient convolutions, reducing the per-inference time to roughly six minutes. Rovida~\emph{et~al.}~\cite{rovida2024encrypted} focused on the memory
usage and proposed an optimized vector encoding strategy, significantly decreasing the number of stored rotation keys required by HE. Notably, They allows the evaluation of an encrypted ResNet20 in around five minutes on a laptop by using 15GB of memory.
EPCNN~\cite{he2025epcnn} presented an interactive private inference framework and allowed the client to participate in evaluating non-linear activations and judge the sign of blinded convolution results. Their approach can improve efficiency with minimal client-server interaction.

Different from these HE-only studies, our proposed Pacifo shifts emphasis to \emph{high-level} circuit design for improving efficiency: it uses a layer-freezing-based fine-tuning strategy to adaptively trims the portion of the model evaluated under HE, thereby fundamentally reducing the expensive computations in the context of HE. 

\begin{figure}[t]
    \centering
\includegraphics[width=.9\linewidth]{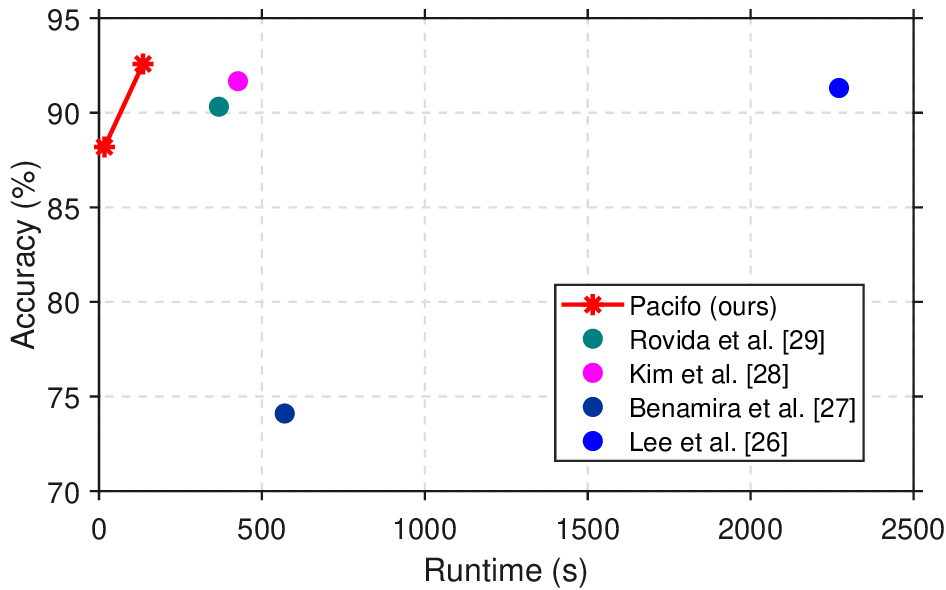}
    \caption{Performance comparison between existing HE-based methods and ours on ResNet20 using the CIFAR-10 dataset.}
    \label{fig:comparison}
\end{figure}

\subsection{Alternative Approaches to Private Inference}

Beyond HE-only solutions, some recent studies~\cite{juvekar2018gazelle,lehmkuhl2021muse,liu2022securely,yang2023efficient,zhao2024pratical} using alternative techniques, e.g., Garbled Circuit (GC) and Secret Sharing (SS), for private inference. These approaches are typically \emph{interactive}, requiring the server and client to collaborate online during evaluation.

Sonic~\cite{liu2022securely} adopted an SS-based two-server setting, secret-sharing both user inputs and model parameters and assuming the two servers do not collude (honest-majority). 
This removes the need for the user to stay online during inference process, but the security and correctness are fragile to server collusion or compromise. Zhao~\emph{et~al.}~\cite{zhao2024pratical} also used an SS-based dual-server setting. They optimized the secure convolution protocol, and further designed several lightweight secure building blocks, hence effectively improving the communication efficiency.
Besides, some studies proposed hybrid frameworks to combine the strengths of different schemes. For instance, Gazelle~\cite{juvekar2018gazelle} used HE for the linear layers of CNN and employed GC for the non-linear ones, while the communication between the client and server is supported by SS. This brings Gazelle a high computation efficiency in a typical single-server setting with limited communicatio costs. Similarly, MUSE~\cite{lehmkuhl2021muse} combined HE,GC and SS and designed an efficient two-party protocol resilient to malicious clients. It enables efficient switching between authenticated additive shares and GC labels via conditional disclosure, and accelerates Beaver triple generation in SS. 

Notably, the above solutions all work in an interactive manner, and some of them~\cite{liu2022securely,yang2023efficient,zhao2024pratical} require multiple servers, which is significantly different from our design of Pacifo. Compared with them, the non-interactivity of Pacifo allows minimal communication overhead, and aligns naturally with client–server deployments such as MLaaS.






\section{Preliminaries}\label{Sec:preliminary}

\subsection{Notations.}
We use symbols $\mathbb{R}$ and $\mathbb{Z}$ to denote sets of real numbers and integers, respectively.
The uppercase bold letter $\boldsymbol{X}\in \mathbb{R}^{a \times b}$ is used for a matrix with $a\times b$ real values, and $\boldsymbol{X}[i,j]$ denotes the element indexed by row $i$ and column $j$ of $\boldsymbol{X}$. The 
bold lowercase letter $\boldsymbol{x}$ is for vectors, and $x_i$ denotes its $i$-th element.
The lowercase letter $x$ is used for scalar values. We write $[n]$ to represent the set $\{0,1,\ldots,n-1\}$. A symbol with a tilde, e.g., $\tilde{\boldsymbol{x}}$, denotes the encrypted result of $\boldsymbol{x}$ by HE.

\subsection{Threat Model}
Following related studies, e.g.,~\cite{hao2022iron,pang2024bolt,dong2023puma,hou2023ciphergpt}, we adopt the semi-honest (i.e., honest-but-curious) adversarial model. In this setting, the parties involved adhere to the established protocols faithfully but may attempt to extract additional private information passively. This assumption is practical, as a foundational level of mutual trust is necessary for implementing inference services. Malicious deviations from the inference protocol would lead to inference failures, offering no benefit to either the client or the server.

\subsection{ResNet20 Model.}\label{Subsec:resnet20}
ResNet20 is a convolutional neural network with residual connections to address the vanishing gradient problem~\cite{he2016deep}. Fig.~\ref{fig:resnet20} presents the architecture of a typical ResNet20 model and highlights its layer categorizations. The model primarily consists of nine residual blocks, each comprising convolutional layers, batch normalization, and ReLU activations.\\
\textbf{Residual Block.} The residual block in ResNet20 introduces shortcut connections to preserve input information for deeper layers. A basic residual block is presented on the right side of Fig.~\ref{fig:resnet20}, and it could be defined as follows:
\begin{align}
    \mathrm{ResidualBlock}(\boldsymbol{X})=\mathrm{ReLU}(\boldsymbol{X}+\mathrm{F}(\boldsymbol{X})),
\end{align}
where $\mathrm{F}(\boldsymbol{X})$ represents the transformation applied on input $\boldsymbol{X}$, consisting of two convolutional layers, batch normalization, and ReLU activation. Specifically, it is expressed as follows:
\begin{align}
    \mathrm{F}(\boldsymbol{X})=\mathrm{ConvBN}(\mathrm{ReLU}(\mathrm{ConvBN}(\boldsymbol{X},\boldsymbol{W}_1)),\boldsymbol{W}_2),
\end{align}
where $\mathrm{ConvBN}(\cdot)$ denotes a convolution operation followed by batch normalization for simplicity; $\boldsymbol{W}_1$ and $\boldsymbol{W}_2$ denote the corresponding convolution parameters.
\\

\begin{figure}[t]
    \centering
\includegraphics[width=.9\linewidth]{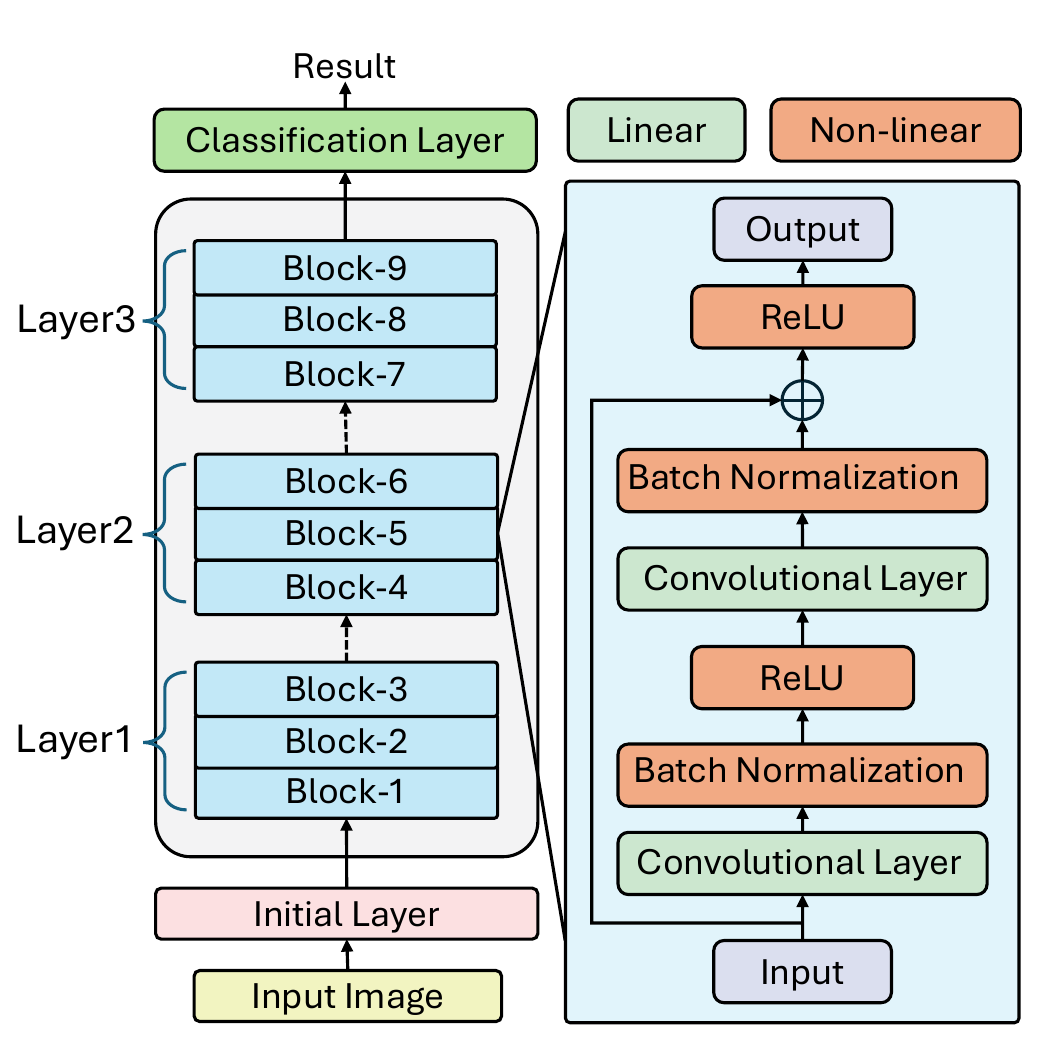}
    \caption{The architecture of a ResNet20 model.}
    \label{fig:resnet20}
\end{figure}

\subsection{RNS-CKKS Scheme}\label{subsec:HE}
We employ \textit{residue number system} (RNS) variant of Cheon-Kim-Kim-Song (CKKS) in this paper, which is a \textit{leveled} FHE scheme supporting a multiplicative depth of $L$. Once the level of a ciphertext becomes too high, a \textit{bootstrapping} operation is required to refresh it for more computations.
A summary of basic homomorphic operations used in this paper is as follows, and we employ them in a black-box manner:
\begin{itemize}
    \item $p(x)=\Pi_{\textnormal{Code}}(\boldsymbol{x})$. The encoding protocol $\Pi_{\textnormal{Code}}$ inputs a vector $\boldsymbol{x}$ and outputs a polynomial $p(x)$.
    \item $\hat{x}=\Pi_{\textnormal{Enc}}(\boldsymbol{x})$. The encryption protocol $\Pi_{\textnormal{Enc}}$ inputs a vector $\boldsymbol{x}$ and outputs a SIMD ciphertext $\hat{x}$.
    \item $\boldsymbol{x}=\Pi_{\textnormal{Dec}}(\hat{x})$. The decryption protocol inputs a SIMD ciphertext $\hat{x}$ and outputs a vector $\boldsymbol{x}$.
    \item $\hat{c}=\Pi_{\textnormal{Add}}(\hat{a},b)$. The ciphertext-plaintext addition protocol inputs a ciphertext $\hat{a}$ and a plaintext $b$; outputs the result: $\Pi_{\textnormal{Enc}}([a_0+b_0,\ldots,a_{N-1}+b_{N-1}])$.
    \item $\hat{c}=\Pi_{\textnormal{Add}}(\hat{a},\hat{b})$. The ciphertext-ciphertext addition protocol inputs two ciphertexts $\hat{a}$ and $\hat{b}$; outputs the result: $\Pi_{\textnormal{Enc}}([a_0+b_0,\ldots,a_{N-1}+b_{N-1}])$.
    \item $\hat{c}=\Pi_{\textnormal{Mul}}(\hat{a},b)$. The ciphertext-plaintext multiplication protocol inputs a ciphertext $\hat{a}$ and a plaintext $b$; outputs the result: $\Pi_{\textnormal{Enc}}([a_0b_0,\ldots,a_{N-1}b_{N-1}])$.
    \item $\hat{c}=\Pi_{\textnormal{Mul}}(\hat{a},\hat{b})$. The ciphertext-ciphertext multiplication protocol inputs two ciphertexts $\hat{a}$ and $\hat{b}$; outputs the result: $\Pi_{\textnormal{Enc}}([a_0b_0,\ldots,a_{N-1}b_{N-1}])$.
    \item $\hat{x}=\Pi_{\textnormal{Boot}}(\hat{x})$. The bootstrapping protocol $\Pi_{\textnormal{Boot}}$ inputs a ciphertext $\hat{x}$ and outputs a refreshed ciphertext with reduced noise. This operation is naturally supported in the OpenFHE library~\cite{OpenFHE}.
\end{itemize}









\section{Basic Circuit Design}\label{Sec:basic_design}

This section provides an overview of our circuit design for ResNet20. Specifically, Section~\ref{subsec:convolution} details the convolution design; Section~\ref{subsec:batchnorm} describes the implementation of batch normalization; Section~\ref{subsec:relu} discusses the ReLU implementation; Section~\ref{subsec:bootstrapping} covers bootstrapping; and Section~\ref{subsec:overallcircuit} summarizes the overall circuit architecture.

\subsection{Convolution}\label{subsec:convolution}

\begin{figure*}[h]
    \centering
\includegraphics[width=.95\linewidth]{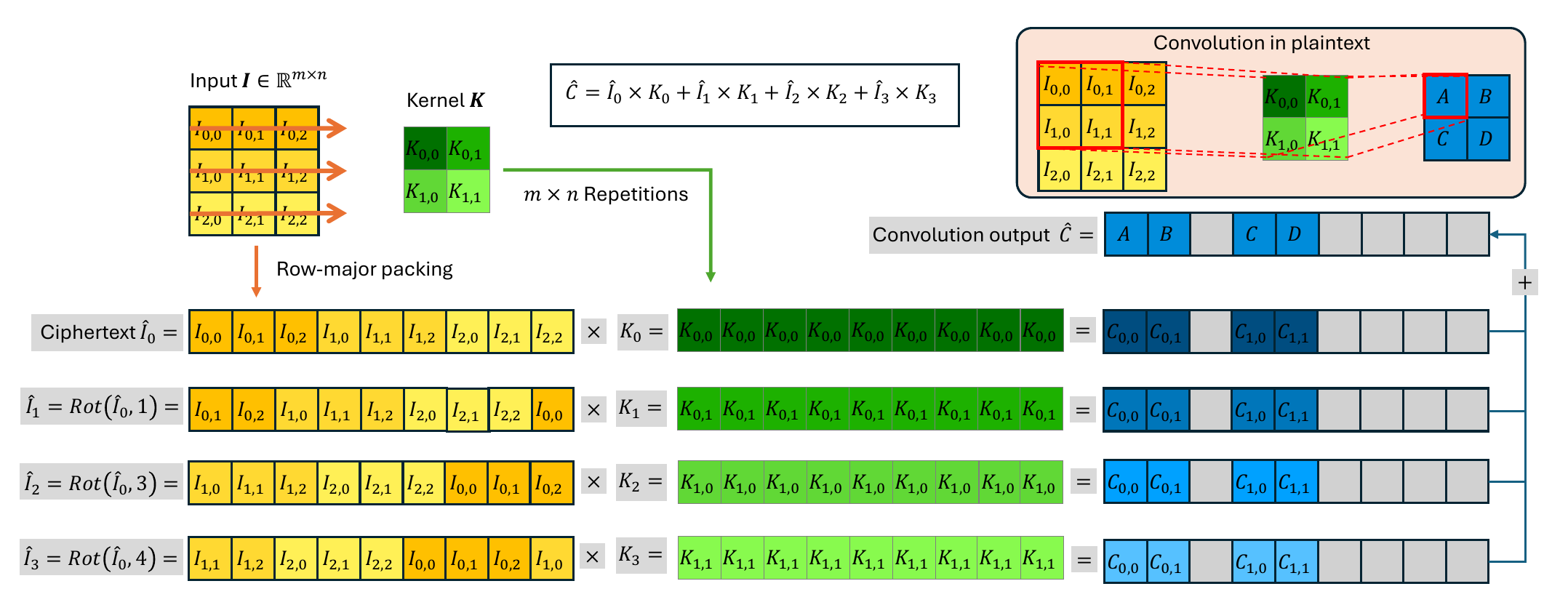}
    \caption{A toy example of our utilized vector encoding method for 2D convolutional computation with an input feature map $I$ of size $3\times 3$ and a kernel $K$ of size $k=2$. Specifically, the input $I$ needs to be packed and encrypted before computation.}
    \label{fig:convolution}
\end{figure*}

We utilize the \textit{Vector Encoding} approach~\cite{juvekar2018gazelle,dathathri2019chet,dathathri2020eva} to facilitate the convolutional computation in the context of FHE. Let $\boldsymbol{I}\in\mathbb{R}^{c_{\textnormal{in}}\times w \times w}$ denote an input feature map and assume $c_{\textnormal{in}}\cdot w\cdot w\leq N/2$, we are allowed to encrypt $\boldsymbol{I}$ into a single ciphertext $\hat{C}$ in a row-major fashion as follows:
\begin{align}
    \hat{C}= \Pi_{\textnormal{Enc}}([c_n\mid_{n \in [c_{\textnormal{in}}w^2]}]),~&\textnormal{s.t.}~c_{w^2i+wi+k} = \boldsymbol{I}[i,j,k],\nonumber\\
    &\textnormal{for}~i\in [c_{\textnormal{in}}],j \in [w], k \in [w].\label{Equa:input}
\end{align}
For a plaintext filter $\boldsymbol{F}\in \mathbb{R}^{c_{\textnormal{in}}\times k \times k}$, we need to create $k^2$ plaintext vectors containing $w^2$ repetitions of the corresponding kernel value.
For instance, a plaintext vector encoding the kernel value $\boldsymbol{F}[c_{\textnormal{in}},m,n]$ can be formally defined as follows:
\begin{align}
    \boldsymbol{h}_{m,n}\!= \![\boldsymbol{h}_i\!\mid_{i \in [c_{\textnormal{in}}]}] ,\textnormal{where}~\boldsymbol{h}_i\!=&[\underbrace{\boldsymbol{F}[i,m,n],\!\dots\!,\boldsymbol{F}[i,m,n]}_{w^2 \text{ times}}],\nonumber\\
    &\textnormal{for}~m \in [k], n \in [k],\label{Equa:kernel1}
\end{align}
and the final output $\boldsymbol{H}$ containing $k^2$ vectors is defined as: 
\begin{align}
    &\boldsymbol{H}\!=\! [\boldsymbol{h}_{0,0},\!\ldots\!,\boldsymbol{h}_{m,n},\!\ldots\!,\boldsymbol{h}_{k-1,k-1}]~\textnormal{for}~m \in [k], n \in [k].\label{Equa:kernel2}
\end{align}
With the obtained $k^2$ plaintext kernel vector, we could perform rotations on the ciphertext $\hat{C}$ to create another $k^2-1$ ciphertext for the alignment. Particularly, the required rotations on the input are dependent on the size of the kernels~\cite{kim2023optimized,rovida2024encrypted}.
After creating a set of input rotations $\hat{\boldsymbol{C}}=[\hat{C}_{0,0},\ldots,\hat{C}_{k-1,k-1}]$, we are able to align them with their corresponding plaintext kernel vectors and perform the following FHE computations:
\begin{align}
    \hat{O}&= \Pi_{\textnormal{Add}}(\Pi_{\textnormal{Mul}}(\hat{\boldsymbol{C}},\boldsymbol{H}))\nonumber\\
    &=\Pi_{\textnormal{Add}}(\underbrace{\Pi_{\textnormal{Mul}}(\hat{C}_{0,0},\boldsymbol{h}_{0,0}),\!\ldots\!,\Pi_{\textnormal{Mul}}(\hat{C}_{k-1,k-1},\boldsymbol{h}_{k-1,k-1})}_{k^2 \text{ multiplications}}), \label{Equa:cnnoutput}
\end{align}
where $\hat{O}$ denotes the output ciphertext containing the convolution result. Until now, we have successfully applied a filter containing $c_{\textnormal{in}}$ filters of size $k \times k$ over an input feature map of $\mathbb{R}^{c_{\textnormal{in}}\times w \times w}$. 
The whole process above requires $k^2-1$ rotations and $k^2$ multiplications. To better illustrate the whole process, we present a toy example in Fig.~\ref{fig:convolution}, where a kernel of $\mathbb{R}^{1\times 2 \times 2}$ is applied on an input feature map of $\mathbb{R}^{1\times 3 \times 3}$. 

Most of the convolutions in ResNet20 (excluding downsampling) are evaluated with kernels of size $\mathbb{R}^{c_{\textnormal{in}}\times 3 \times 3}$ on an input map of size $\mathbb{R}^{c_{\textnormal{in}}\times 32 \times 32}$, and our described vector encoding approach can be easily extended to this setting~\cite{lee2022low, rovida2024encrypted}. 
Moreover, it is worth noting that a complete convolutional layer usually contains $c_{\textnormal{out}}$ filters, and we could denote the $i$-th filter as $\boldsymbol{F}_i$, for $i\in[c_{\textnormal{out}}]$. The final convolution result $\hat{R}$ is obtained by aggregating the results of applying all filters $\{\boldsymbol{F}_i\}_{i\in[c_{\textnormal{out}}]}$ on the input feature map. We present the complete secure convolution process in Algorithm~\ref{Algorithm:Convolution}. For the downsampling operations in ResNet20, we follow Rovida~\emph{et~al.}~\cite{rovida2024encrypted} and implement them by rearranging values in the output ciphertext through additional masking and rotations. \\
\textbf{Computational Complexity.} According to Algorithm~\ref{Algorithm:Convolution}, we can calculate that there are $k^2\cdot c_{\textnormal{out}}$ multiplications, $k^2+c_{\textnormal{out}}-1$ rotations, and $k^2\cdot c_{\textnormal{out}}$ additions, which means the overall complexity is only related to kernel size and output channels. 





    
    

\begin{algorithm}
\caption{Secure Convolution} 
\label{Algorithm:Convolution}
\KwIn{
Encrypted feature map $\hat{C}$ (cf. Equation (\ref{Equa:input})), rotation indexes $\boldsymbol{I}$, 
set of filters $\{\boldsymbol{F}_i\}_{i\in[c_{\textnormal{out}}]}$, input feature map width $w$
}
\KwOut{Encrypted result $\hat{R}$}

\SetAlgoLined
\DontPrintSemicolon
\SetKwProg{Fn}{Function}{:}{\textbf{End function}}

\Fn{$\mathrm{Convolution}(\hat{C},\boldsymbol{I},\{\boldsymbol{F}_i\}_{i\in[c_{\textnormal{out}}]})$}{
    


    
 \tcp{Apply rotations on the input feature map}
 $\hat{\boldsymbol{C}}\leftarrow \{\}$ 
 
    \For{each $\mathrm{index}$ in $\boldsymbol{I}$}{
    $\hat{\boldsymbol{C}}_i \leftarrow \Pi_{\textnormal{Rot}}(\hat{C},\mathrm{index})$
    }

\tcp{Apply filters on the input rotations}
    $\hat{R} \leftarrow \Pi_{\textnormal{Enc}}(0)$ 

    \For{$i \leftarrow 0$ to $c_{\textnormal{out}}$}{

    $\boldsymbol{H} \leftarrow \mathrm{EncodeFilter}(\boldsymbol{F}_i)$ (cf. Equations (\ref{Equa:kernel1}), (\ref{Equa:kernel2}))

    $\hat{O}= \Pi_{\textnormal{Add}}(\Pi_{\textnormal{Mul}}(\hat{\boldsymbol{C}},\boldsymbol{H}))$ (cf. Equation (\ref{Equa:cnnoutput}))
    

     $\hat{R}\leftarrow \Pi_{\textnormal{Add}}(\hat{R},\Pi_{\textnormal{Rot}}(\hat{O},w^2))$
    }

    \textbf{Return} $\hat{R}$
}
\end{algorithm}






\subsection{Batch Normalization}\label{subsec:batchnorm}

For a given feature map $\boldsymbol{x}$ in a mini-batch, BatchNorm applies the following transformation on each element $x_i$:
\begin{align}
    y_i=\frac{x_i-\mathrm{E}(x)}{\sqrt{\mathrm{Var}(x)+\epsilon}} \cdot \gamma + \beta,
    \label{Equa:BatchNorm}
\end{align}
where $\mathrm{E}(x)$ and $\mathrm{Var}(x)$ are the mean and variance computed over the current mini-batch, $\epsilon$ is a small constant to prevent division by zero, and $\gamma$ and $\beta$ are trainable affine parameters. Our starting point is that when doing inference, values of $\mathrm{E}(x)$ and $\mathrm{Var}(x)$ are actually pre-computed and can be regarded as constants. In this case, we can simplify the computation in Equation~(\ref{Equa:BatchNorm}) by reformulating the coefficients as follows:
\begin{align}
    y_i&=\frac{x_i-E}{V} \cdot \gamma + \beta,\\
    &= \frac{\gamma}{V}x_i - \frac{E\gamma}{V} + \beta.
    \label{Equa:Simplified_BatchNorm}
\end{align}
It is obvious that Equation~(\ref{Equa:Simplified_BatchNorm}) can be viewed as an affine transformation: $Ax_i+B$, where $A=\frac{\gamma}{V}$ and $B=- \frac{E\gamma}{V} + \beta$, and its computation consumes $1$ circuit depth. However, we can combine BatchNorm with the former convolution to save the depth consumption further. Algorithm~\ref{Algorithm:ConvBN} describes our method for integrating BatchNorm with convolution computations.

\begin{algorithm}
\caption{Convolution with BatchNorm (ConvBN)} 
\label{Algorithm:ConvBN}
\KwIn{
Encrypted input feature map $\hat{C}$, rotation indexes $\boldsymbol{I}$, 
set of filters $\{\boldsymbol{F}_i\}_{i\in[c_{\textnormal{out}}]}$, 
BN parameters $\{\gamma_i, \beta_i, E_i, V_i\}_{i\in[c_{\textnormal{out}}]}$, input width $w$
}
\KwOut{Encrypted output $\hat{R}$ after Conv+BN}

\SetAlgoLined
\DontPrintSemicolon
\SetKwProg{Fn}{Function}{:}{\textbf{End function}}

\Fn{$\mathrm{ConvBN}(\hat{C}, \boldsymbol{I}, \{\boldsymbol{F}_i\}, \{\gamma_i, \beta_i, E_i, V_i\}, w)$}{
    
    $\hat{\boldsymbol{C}} \leftarrow \{\}$\;
    \For{each $\textnormal{index}$ in $\boldsymbol{I}$}{
        $\hat{\boldsymbol{C}}_{\textnormal{index}} \leftarrow \Pi_{\textnormal{Rot}}(\hat{C}, \textnormal{index})$\;
    }

    $\hat{R} \leftarrow \Pi_{\textnormal{Enc}}(0)$

    \For{$i \leftarrow 0$ to $c_{\textnormal{out}}-1$}{
        $A_i \leftarrow \frac{\gamma_i}{V_i}$\;
        $B_i \leftarrow -\frac{\gamma_i \cdot E_i}{V_i} + \beta_i$\;

        $\boldsymbol{H} \leftarrow \mathrm{EncodeFilter}(A_i \cdot \boldsymbol{F}_i)$ \tcp{Scale filter}

        $\hat{O} \leftarrow \Pi_{\textnormal{Add}}(\Pi_{\textnormal{Mul}}(\hat{\boldsymbol{C}}, \boldsymbol{H}))$ 

        $\hat{B} \leftarrow \Pi_{\textnormal{Enc}}(B_i)$ \tcp{Encode bias}

        $\hat{O} \leftarrow \Pi_{\textnormal{Add}}(\hat{O}, \hat{B})$ \tcp{Add BN bias}

        $\hat{R} \leftarrow \Pi_{\textnormal{Add}}(\hat{R}, \Pi_{\textnormal{Rot}}(\hat{O}, w^2))$\;
    }

    \textbf{Return} $\hat{R}$
}
\end{algorithm}

\subsection{ReLU}\label{subsec:relu}
ReLU is a widely used activation function that introduces non-linearity in neural networks.
For an input value $x$, the ReLU activation is defined as follows:
\begin{align}
    \mathrm{ReLU}(x)=\max{(0,x)},
\end{align}
As FHE only supports polynomial operations, it is nontrivial to compute $\max{(0,x)}$ in a non-interactive manner. Previous works~\cite{lee2022privacy,rovida2024encrypted} usually approximated ReLU with a large-degree polynomial to maintain its precision. However, while this method maintains function accuracy, it often consumes many FHE level budgets, and a large running time is required. In this paper, we provide an aggressive but much more efficient approximation for ReLU.
Specifically, we first experimentally observe the range of input values (i.e, $[-\alpha,\alpha]$) to ReLU with the plaintext model. Then, we scale all the input into the range of $[-1,1]$ by setting a scale $1/\alpha$ in the previous convolutional layer. By doing so, we are able to use an aggressive quadratic to efficiently approximate ReLU within the range of $[-1,1]$:
\begin{align}
    \mathrm{ReLU}(x)\approx ax^2+bx+c, ~\textnormal{for}~\abs{x} \leq 1,
    \label{approx:relu}
\end{align}
where we use the Remez exchange algorithm~\cite{remez1934determination} to determine the optimal coefficients $a$, $b$, and $c$. This iterative method minimizes the maximum error between the polynomial and the target function over a given interval. The obtained results are $a=0.4950$, $b=0.5000$, $c=0.0656$. However, since we approximate ReLU using only a quadratic polynomial in (\ref{approx:relu}), the maximum error remains $6.06\times10^{-2}$. There are two main reasons why we could use this aggressive approximation:
\begin{itemize}
    \item With layer freezing optimization, most ReLU activations can be computed accurately in plaintext on the client side, with only a small subset requiring the approximation.
    \item We utilize knowledge distillation to effectively mitigate the performance drop caused by approximated ReLUs, ensuring minimal impact on overall model accuracy.
\end{itemize}
\textbf{Efficient Polynomial Computation.} We could utilize the standard Horner's method~\cite{wang1955horner} to optimize the cost of evaluating equation (\ref{approx:relu}) further. Specifically, we first evaluate $(ax+b)$, and then compute $(ax+b)\cdot x+c$, saving one multiplication compared to the naive method $ax^2+bx+c$. Algorithm~\ref{Algorithm:ReLU} presents the details of our secure ReLU evaluation.
\begin{algorithm}
\caption{Secure ReLU} 
\label{Algorithm:ReLU}
\KwIn{
$\hat{\boldsymbol{x}}=\Pi_{\textnormal{Enc}}([x_0,\ldots,x_{n-1}])$; 
}
\KwOut{$\hat{\boldsymbol{y}}=\Pi_{\textnormal{Enc}}([y_0,\ldots,y_{n-1}])$ (cf. Equation (\ref{approx:relu}))}

\SetKwFunction{FReLU}{ReLU}
\SetAlgoLined
\DontPrintSemicolon
\SetKwProg{Fn}{Function}{:}{\textbf{End function}}

\Fn{$\mathrm{ReLU}(\hat{\boldsymbol{x}})$}{
    $\hat{\boldsymbol{s}} \leftarrow =\Pi_{\textnormal{Add}}(\Pi_{\textnormal{Mul}}(\hat{\boldsymbol{x}},a),b)$ \tcp{$ax+b$}
    
    $\hat{\boldsymbol{y}} \leftarrow =\Pi_{\textnormal{Add}}(\Pi_{\textnormal{Mul}}(\hat{\boldsymbol{s}},\hat{\boldsymbol{x}}),c)$ \tcp{$(ax+b)x+c$}
    
    \textbf{Return} $\hat{\boldsymbol{y}}$
}
\end{algorithm}
\subsection{Design of bootstrapping}\label{subsec:bootstrapping}

Once a ciphertext’s level becomes too low, a bootstrapping operation is required to restore it to a higher level, enabling further computations. However, the operation itself also consumes a certain number of levels.
The bootstrapping operation is often the most expensive in FHE-based transformers. For instance, it accounts for $37.72\%$ and $53.96\%$ of the total runtime in NEXUS~\cite{zhang2024secure} and THOR~\cite{moon2024thor}, respectively. Therefore, careful design and strategic placement of bootstrapping operations are crucial for optimizing runtime performance.

In this work, we use the bootstrapping scheme proposed by Bossuat~\textit{et al.}~\cite{bossuat2021efficient}. Our goal is to increase the efficiency of each bootstrapping while reducing the number of times it is invoked.
Specifically, we use the sparse ternary scheme~\cite{cheon2017homomorphic} for secret key distribution, which provides lower computational complexity and higher precision. For the level budgets required for encoding and decoding operations in bootstrapping, i.e., CoeffToSlot (CtoS) and SlotToCoeff (StoC), we allocate them as $\{3,3\}$ to ensure high computational efficiency. Additionally, an extra $9$ levels are required for ModRaise, Chebyshev approximations, and double-angle iterations during bootstrapping, bringing the total level budget for a single bootstrapping operation to $15$. By carefully setting the multiplicative depth, we limit the number of bootstrapping operations to just two per residual block. We will further analyze the impact of each bootstrapping operation on accuracy and computational overhead in Section~\ref{subsec:microbenchmarks}.

\begin{figure}[h]
    \centering
\includegraphics[width=.8\linewidth]{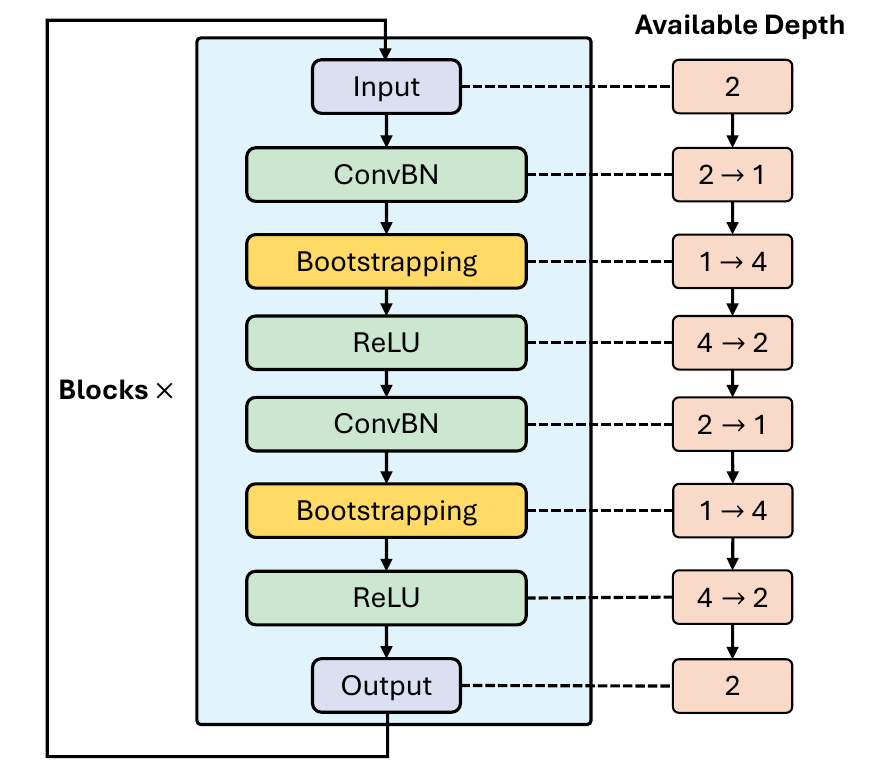}
    \caption{Placement of bootstrappings in a block of ResNet20.}
    \label{fig:circuit}
\end{figure}

\subsection{Overall Circuit}\label{subsec:overallcircuit}

In this subsection, we discuss our FHE circuit designed for a basic residual block in ResNet20. Specifically, Algorithm~\ref{Algorithm:residualblock} generalizes the FHE operations required to evaluate a block, which is also aligned with the process in Fig.~\ref{fig:circuit}. Given that ResNet20 is composed of multiple residual blocks stacked sequentially, our design for an individual block can be seamlessly extended to the entire model, as illustrated in Fig.~\ref{fig:resnet20}.

\begin{algorithm}
\caption{Secure ResidualBlcok} 
\label{Algorithm:residualblock}
\KwIn{
$\hat{\boldsymbol{x}}=\Pi_{\textnormal{Enc}}([x_0,\ldots,x_{n-1}])$; 
}
\KwOut{$\hat{\boldsymbol{y}}=\Pi_{\textnormal{Enc}}([y_0,\ldots,y_{n-1}])$ (cf. Equation (\ref{approx:relu}))}

\SetAlgoLined
\DontPrintSemicolon
\SetKwProg{Fn}{Function}{:}{\textbf{End function}}

\Fn{$\mathrm{ResidualBlock}(\hat{\boldsymbol{x}})$}{
    $\hat{\boldsymbol{t}} \leftarrow \mathrm{ConvBN}(\hat{\boldsymbol{x}})$ \tcp{Convolution + BatchNorm}

    $\hat{\boldsymbol{t}} \leftarrow \Pi_{\textnormal{Boot}}(\hat{\boldsymbol{t}})$
    
    $\hat{\boldsymbol{t}} \leftarrow \mathrm{ReLU}(\hat{\boldsymbol{t}})$ 

    $\hat{\boldsymbol{t}} \leftarrow \mathrm{ConvBN}(\hat{\boldsymbol{t}})$

    $\hat{\boldsymbol{t}} \leftarrow \Pi_{\textnormal{Boot}}(\hat{\boldsymbol{t}})$ 

    $\hat{\boldsymbol{t}} \leftarrow \mathrm{ReLU}(\hat{\boldsymbol{t}})$ 

    $\hat{\boldsymbol{y}} \leftarrow \Pi_{\textnormal{Add}}(\hat{\boldsymbol{t}},\hat{\boldsymbol{x}})$ 
    \tcp{Residual connection}

    \textbf{Return} $\hat{\boldsymbol{y}}$
}
\end{algorithm}




\section{Fine-Tuning based Layer Partition}\label{Sec:layer_freezing}

In this section, we describe our approach to reducing the reliance on expensive HE operations from a model architecture perspective. Section~\ref{subsec:layer_partition} introduces a layer partitioning strategy designed to minimize the number of model layers evaluated under HE. Section~\ref{subsec:TL} further details how we apply transfer learning to maintain model accuracy after partitioning.

\subsection{Layer Partition for Efficiency}\label{subsec:layer_partition}

Our starting point is that when an open-source, pre-trained model is fine-tuned on a downstream dataset, updating only a subset of model parameters can achieve good performance. In such cases, only the fine-tuned layers need to remain private during inference, while the unchanged layers can be deemed as ``public'' and shared without compromising model privacy. Consequently, leveraging these public layers appropriately can lead to improvements in computational efficiency.
This insight is also leveraged in prior works~\cite{rathee2024mpc,lam2024efficient}. However, \cite{rathee2024mpc} applies this idea solely within the context of MPC, while \cite{lam2024efficient} adds extra network layers and evaluates on a relatively simple binary image classification task (Cat vs Dog).

Recall from Section~\ref{subsec:HE} that a key property of HE is that any computation involving ciphertext inputs always produces a ciphertext output. Hence, to reduce the number of operations under HE, one must separate the fine-tuned layers (i.e., private layers) from the public layers. Specifically, we could restrict fine-tuning updates to the final layers of the model, with only these layers needing to be protected and retained on the server, which is as shown in Fig.~\ref{fig:layer_partition}. 
Let $\mathcal{F}_{\textnormal{pub}}$ and $\mathcal{F}_{\textnormal{priv}}$ denote the public and private portions of the model, respectively. The client first performs inference over the public layers to obtain an intermediate representation $\mathbf{z}_{\textnormal{pub}} = \mathcal{F}_{\textnormal{pub}}(\mathbf{X})$, which is then encrypted into $\hat{z}_{\textnormal{pub}}$ and sent to the server. The server continues inference using the private layers, producing the encrypted prediction $\hat{\mathbf{y}} = \mathcal{F}_{\textnormal{priv}}(\hat{\mathbf{z}}_{\textnormal{pub}})$. Finally, the client decrypts $\hat{\mathbf{y}}$ to obtain the prediction $\mathbf{y}$. The whole process is also generalized in Algorithm~\ref{Algorithm:inference}.
Consequently, assume $\mathcal{F}_{\textnormal{pub}}$ and $\mathcal{F}_{\textnormal{priv}}$ consist of $r$ and $T-r$ residual blocks, respectively. Then, under the assumption that each residual block incurs a similar HE cost, the overall HE computation overhead is reduced by a factor of $\frac{T}{T-r}$, effectively saving $\frac{r}{T} \times100\%$ of homomorphic computation. In practice, since each block typically involves bootstrapping and nonlinear approximation steps, the reduction in latency could be higher than the linear estimate (Section~\ref{subsec:macrobenchmarks}).

\begin{algorithm}
\caption{Private Inference via Layer Partitioning}
\label{Algorithm:inference}
\KwIn{Input image $\mathbf{X} \in \mathbb{R}^{C \times H \times W}$;}
\KwOut{Predicted label probabilities $\boldsymbol{y} \in \mathbb{R}^{K}$;}

\SetAlgoLined
\DontPrintSemicolon

\textbf{Let:} \\
\Indp
    $\mathcal{F}\! =\! \{L_1, \ldots, L_T\}$ denote the full ResNet20 model, where each $L_i$ is a ResidualBlock;\\
    $\mathcal{F}_{\text{pub}} \!=\! \{L_1, \ldots, L_r\}$: public layers (client side); \\
    $\mathcal{F}_{\text{priv}}\!=\! \{L_{r+1}, \ldots, L_T\}$: private layers (server side); \\
    $\mathcal{F}_{\text{class}}$: final classifier layer (server side); \\
\Indm

\vspace{0.5em}

\textbf{Client executes:} \\
\Indp
    $\mathbf{z}_{\textnormal{pub}} = \mathcal{F}_{\textnormal{pub}}(\mathbf{X})$; \tcp{Apply public layers} 
    
    $\hat{z}_{\textnormal{pub}} \leftarrow \Pi_{\textnormal{Enc}}(\mathbf{z}_{\textnormal{pub}})$; 
    
    Send ciphertext $\hat{z}_{\textnormal{pub}}$ to the server; \\
\Indm

\vspace{0.5em}

\textbf{Server executes:} \\
\Indp
    $\hat{z}_{\textnormal{priv}} = \mathcal{F}_{\textnormal{priv}}(\hat{z}_{\textnormal{pub}})$; \tcp{Apply private layers (cf. Algorithm \ref{Algorithm:residualblock})}

    $\hat{y} = \mathcal{F}_{\textnormal{class}}(\hat{z}_{\textnormal{priv}})$; 
    
    Send ciphertext $\hat{y}$ back to the client; \\
\Indm

\vspace{0.5em}

\textbf{Client executes:} \\
\Indp
    $\mathbf{y} \leftarrow \Pi_{\textnormal{Dec}}(\hat{y})$; \tcp{Decrypt to obtain the result}
\Indm

\end{algorithm}

While delegating the computation of public layers to clients may seem like a limitation, it imposes minimal computational overhead~\cite{rathee2024mpc}. In fact, plaintext inference is less intensive than the encryption and decryption operations that clients must perform regardless. 

\subsection{Transfer learning for Accuracy}\label{subsec:TL}

Although layer partition can significantly reduce the computational overhead of HE, it also introduces constraints on the fine-tuning process, as only the final layers of the model (i.e., the private portion) are permitted to be updated. To ensure high efficiency, the number of private layers is typically much smaller than that of public layers. This limited parameter space restricts the model's capacity to adapt to new tasks or datasets, which could lead to degraded performance.

To mitigate this issue, we adopt the Transfer Learning (TL) technique to keep useful features during pre-training and help the model adapt to the new dataset. Specifically, given a target downstream dataset $\mathcal{D}_{A}$, we begin by selecting an open-source model that has been pre-trained on a related source dataset $\mathcal{D}_{B}$. The source dataset is typically larger or more diverse than the target dataset, while sharing similar data modalities and semantic domains (e.g., CIFAR-10 vs. CIFAR-100). This enables the model to learn general representations from $\mathcal{D}_{B}$ that transfer well to $\mathcal{D}_{A}$. In this case, the pre-trained model parameters from $\mathcal{D}_{B}$ serve as a powerful initialization. Then, during the subsequent fine-tuning process on $\mathcal{D}_{A}$, we could freeze public layers to retain low-level feature extraction capabilities and update only the private layers to adapt high-level features to the target task. This TL-based approach is essential in our layer-partitioned setting: without pre-training, randomly initialized public layers would hinder feature extraction and degrade accuracy, while a strong initialization enables robust fine-tuning performance with limited trainable layers.

In conclusion, layer partition and transfer learning complement each other to improve both the efficiency and accuracy of private inference. Layer partition reduces the computational burden of encrypted operations, while transfer learning preserves model performance under limited trainable parameters. Figure \ref{fig:layer_partition} illustrates how these two techniques can be integrated for fine-tuning on a downstream task. Additional details on our experimental setup are provided in Section~\ref{subsec:models_and_datsets}.

\begin{figure}[t]
    \centering
\includegraphics[width=\linewidth]{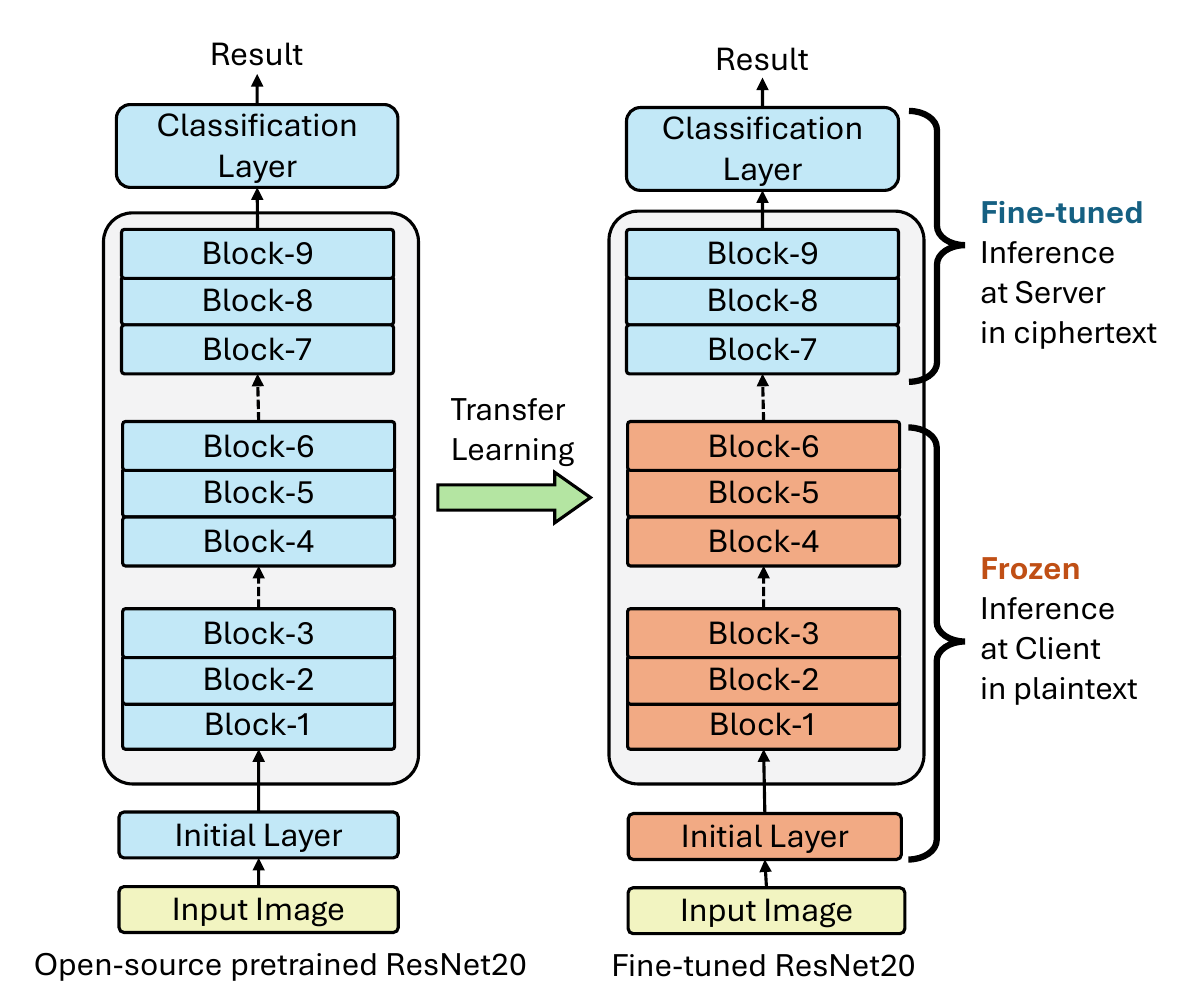}
    \caption{Leverage layer freezing techniques on an open-source pretrained ResNet20 model.}
    \label{fig:layer_partition}
\end{figure}

\section{Customized Security Levels}\label{Sec:security_standards}

In this section, we outline the approach to achieving different security standards. 
In the context of HE, security is typically quantified in terms of \textit{security levels}, which correspond to the computational cost required to breach the cryptographic security~\cite{albrecht2021homomorphic}. A scheme offers $\lambda$-bit security if breaking it requires about $2^{\lambda}$ operations. The most common security levels in HE are $128$-bit, $192$-bit, and $256$-bit. 

Previous related works~\cite{lee2022privacy,lee2022low,kim2023optimized,benamira2023tt,rovida2024encrypted} only support a fixed $128$-bit security, the most widely adopted level due to its balance between efficiency and protection. 
However, $128$-bit security could be inadequate for sensitive or long-term applications, especially with the rise of quantum threats. For instance, lower levels like $80$-bit security have already been deprecated due to vulnerability to modern attacks. Table~\ref{tab:cracking_costs} shows the estimated time required to break different security levels. The estimates are based on the capabilities of different computers as follows:
\begin{itemize}
    \item \textbf{Personal computer:} Capable of ${10}^9$ operations per second, representing standard personal computing devices.
    \item \textbf{Supercomputer:} Capable of $10^{18}$ operations per second, denoting cutting-edge computational systems like Frontier~\cite{rajaraman2023frontier} or Fugaku~\cite{sato2020co}.
    \item \textbf{Quantum computer:} Leverages Grover's algorithm~\cite{grover1996fast} to reduce brute-force complexity to $2^{n/2}$. 
    The time estimates further assume an idealized quantum computer capable of performing $2^{60}$ Grover iterations per second\footnote{This exceeds the capabilities of current or near-term hardware by several orders of magnitude; and it only serves as a theoretical lower bound to illustrate best-case quantum threats.}.
\end{itemize}

\begin{table}[t]
\centering
\caption{Cracking costs for different bit-security levels using various computing systems.}
\label{tab:cracking_costs}
\begin{tabular}{@{}l|lll@{}}
\toprule
\textbf{Security} & \textbf{Personal computer} & \textbf{Supercomputer} & \textbf{Quantum computer} \\ \midrule
80-bit   & \SI{1000}{years} & \SI{14}{days} & \SI{1}{seconds} \\
96-bit   & \SI{1e6}{years} & \SI{2500}{years} & \SI{10}{seconds} \\
128-bit  & \SI{1e20}{years} & \SI{1e13}{years} & \SI{16}{hours} \\
192-bit  & \SI{1e36}{years} & \SI{1e25}{years} & \SI{10}{years} \\
256-bit  & \SI{1e53}{years} & \SI{1e40}{years} & \SI{1e6}{years} \\
\bottomrule
\end{tabular}
\end{table}

To address these evolving needs, our framework is designed to provide flexible support for multiple HE security levels, allowing clients to trade off between performance and security based on their application requirements. Specifically, we have carefully designed three different HE parameter schemes (Table~\ref{Tab:HE_parameters}). And their security levels are estimated using the LWE Estimator~\cite{albrecht2015concrete} with best-known lattice attacks (e.g., primal, dual, and hybrid attacks). 

When configuring these parameters, we ensure that all encryption schemes used in our framework meet at least 128-bit security under the standard RLWE assumption. For use cases requiring higher assurances, we provide parameter sets supporting 192-bit and 256-bit security with adjusted ring dimensions and modulus sizes, taking into account the increased computational overhead.

\begin{table}[h]
  \caption{Key parameters for our RNS-CKKS schemes to achieve different HE security levels.}
  \label{Tab:HE_parameters}
  \centering
  \resizebox{\linewidth}{!}{ 
  \begin{tabular}{c|ccc|c}
    \toprule
    \textbf{Scheme} & \makecell{\textbf{Ring}\\($N$)}
    &\makecell{\textbf{Moduli chain}\\$log_2(PQ)$} &\makecell{\textbf{Noise distribution}\\($\chi_e$)}&\makecell{\textbf{Security level}\\($\lambda$)}   \\\midrule
     $1$& $2^{16}$ &$1228$ bits &$\mathcal{N}_{\mathbb{Z}}(0,\,3.19^2)$ &$128$-bit \\
     $2$& $2^{16}$ &$1188$ bits &$\mathcal{N}_{\mathbb{Z}}(0,\,3.19^2)$ &$192$-bit \\
     $3$& $2^{17}$ &$1289$ bits &$\mathcal{N}_{\mathbb{Z}}(0,\,3.19^2)$ &$256$-bit \\
     \bottomrule
  \end{tabular}
}
\end{table}

\section{Security Analysis}\label{Sec:security_anlaysis}
In this section, we formally analyze the security of the proposed Pacifo. We adopt the simulation-based security paradigm under the static semi-honest adversary model, assuming the adversary corrupts either the client $C$ or the server $S$ during protocol execution.

\subsection{Threat Model and Security Definition}

Following standard definitions, we define the correctness and security of the secure inference protocol as follows:

\begin{definition}[Secure inference protocol]
Let $\Pi$ be an inference protocol executed between a server $S$ holding model parameters $M=(M_1,\dots,M_\ell)$ and a client $C$ holding a private input $A$. The protocol $\Pi$ is secure if the following two properties hold:

\begin{itemize}
    \item \textbf{Correctness:} The output of the protocol execution is always equal to the correct inference result $M(A)$.

    \item \textbf{Security:}
    \begin{itemize}
        \item \textbf{Corrupted Client:} There exists an efficient simulator $\mathsf{Sim}_C$ such that the client’s view $\text{View}^{\Pi}_C$ (consisting of the client's input, randomness, and protocol transcript) is computationally indistinguishable from the simulated view:
        \[
        \text{View}^{\Pi}_C {\approx}_{c} \mathsf{Sim}_C(A,\text{out})
        \]
        where $\text{out}$ denotes the encrypted inference output.

        \item \textbf{Corrupted Server:} There exists an efficient simulator $\mathsf{Sim}_S$ such that the server’s view $\text{View}^{\Pi}_S$ during execution of $\Pi$ is computationally indistinguishable from the simulated view:
        \[
        \text{View}^{\Pi}_S {\approx}_{c} \mathsf{Sim}_S(M)
        \]
    \end{itemize}
\end{itemize}
\end{definition}

\subsection{Security Proof of Pacifo}

Based on this definition, we state the following theorem regarding the security of Pacifo:

\begin{theorem}
Under the semantic security of the underlying RLWE-based FHE scheme, the proposed Pacifo protocol securely realizes the secure inference functionality against a static semi-honest adversary.
\end{theorem}

\begin{proof}
\noindent\textbf{Correctness:}  The correctness of Pacifo directly follows from the correctness of the underlying homomorphic operations. Specifically, due to the correctness of RLWE-based FHE computations, Pacifo always produces the correct inference result in encrypted form, ensuring the protocol satisfies correctness.

\noindent\textbf{Corrupted Client.}  
	The client’s view in the protocol execution consists solely of ciphertexts encrypting the server’s model parameters, denoted as $\{\mathsf{Enc}_S(W_i)\}$. We construct an efficient simulator $\mathsf{Sim}_C$ as follows:
	
	\begin{itemize}
		\item $\mathsf{Sim}_C$ generates ciphertexts encrypting zero ($\mathsf{Enc}_S(0)$) with the same size and structure as real ciphertexts and outputs these simulated ciphertexts to the corrupted client.
	\end{itemize}
	
	Due to the semantic security of the underlying RLWE encryption scheme, ciphertexts encrypting the actual model parameters are computationally indistinguishable from ciphertexts encrypting zeros. Thus, the real client view and simulated view are computationally indistinguishable:
	
	\[
	\{\mathsf{Enc}_S(W_i)\} {\approx}_{c} \{\mathsf{Enc}_S(0)\}
	\]
	
	Hence, security against a corrupted client directly reduces to the semantic security assumption of the underlying RLWE encryption.
	
	\noindent\textbf{Corrupted Server.}  
	The server’s view consists of ciphertexts encrypting the client’s input data $\mathsf{Enc}_C(A)$ and the encrypted inference result $\mathsf{Enc}_C(M(A))$. We construct an efficient simulator $\mathsf{Sim}_S$ as follows:
	
	\begin{itemize}
		\item $\mathsf{Sim}_S$ randomly samples plaintext data $\hat{A}$, encrypts it to obtain simulated ciphertexts $\mathsf{Enc}_C(\hat{A})$, and provides these ciphertexts to the corrupted server.
		\item After receiving the encrypted output $\mathsf{Enc}_C(M(A))$ from the ideal inference functionality, the simulator forwards this ciphertext to the corrupted server.
	\end{itemize}
	
	Given the semantic security of the RLWE-based encryption scheme, ciphertexts encrypting random plaintexts ($\mathsf{Enc}_C(\hat{A})$) are computationally indistinguishable from ciphertexts encrypting the actual client input ($\mathsf{Enc}_C(A)$). Thus, we have:
	
	\[
	\{\mathsf{Enc}_C(\hat{A}), \mathsf{Enc}_C(M(A))\} {\approx}_{c} \{\mathsf{Enc}_C(A), \mathsf{Enc}_C(M(A))\}
	\]
	
	Therefore, the server's view is indistinguishable from the simulated view, ensuring security against a corrupted server.

Since both corrupted views are indistinguishable from their simulated counterparts, Pacifo is secure under the static semi-honest adversary model, relying directly on the semantic security of the underlying RLWE-based FHE scheme.
\end{proof}

\section{Experimental Results}\label{Sec:experiments}

\subsection{Implementation}
We implement Pacifo in C++, utilizing the OpenFHE library~\cite{OpenFHE} for CKKS, and use HEXL~\cite{boemer2021intel} to accelerate OpenFHE on an Intel(R) Xeon(R) Gold 5218R CPU. Initially, we follow the HE standard and employ $128$-bit security. For the ResNet20 layer (1 downsampling block + 2 residual blocks) circuit, we set the polynomial degree to $N=2^{16}$, ciphertext modulus as $1756$ bits, and multiplicative depth to $L=24$ for sufficient computation depth and available bootstrappings.
When freezing most blocks, for a more efficient block circuit, we choose $N=2^{15}$, $797$ bits and $L=17$ for higher efficiency. The first modulus size $q_0$ is configured to be $50$ bits, whereas the scaling modulus sizes $\{q_i \mid i > 1\}$ are set to $46$ bits.

\subsection{Models and Datsets}\label{subsec:models_and_datsets}
We use a standard ResNet-20 and pretrain it on CIFAR-100 for $200$ epochs to learn general representations. We then freeze the first six and eight residual blocks, respectively, to instantiate $\textnormal{Pacifo}_{\{\textnormal{L}\}}$ and $\textnormal{Pacifo}_{\{\textnormal{B}\}}$. Intuitively, compared with $\textnormal{Pacifo}_{\{\textnormal{B}\}}$, $\textnormal{Pacifo}_{\{\textnormal{L}\}}$ retains more trainable parameters and thus higher learning capacity, but it incurs greater computational cost during inference as more parameters must be kept private. Finally, we employ KD using a fully trained ResNet-20 teacher on the target downstream dataset (MNIST or CIFAR-10) to supervise $\textnormal{Pacifo}_{\{\textnormal{L}\}}$ and $\textnormal{Pacifo}_{\{\textnormal{B}\}}$, respectively.

\subsection{Micro-benchmark Evaluation of ReLU}\label{subsec:microbenchmarks}




\textbf{ReLU Approximations.} Recall that a $2$-degree polynomial is used for our ReLU approximations; here we further investigate its computational efficiency and influence on accuracy.
Table~\ref{Tab:relu} compares computational efficiency for different ReLU approximations in the context of FHE. It is worth noting that our ReLU approximation consumes only $2$ circuit depths, much less than the other studies~\cite{lee2022privacy,lee2022low,kim2023optimized,rovida2024encrypted}.

\begin{table}[h]
  \caption{Comparison of the computational efficiency for different ReLU approximations in FHE.}
  \label{Tab:relu}
  \centering
  \resizebox{\linewidth}{!}{ 
  \begin{tabular}{l|lcccccccc}
    \toprule
    \textbf{System}  &\textbf{Method} &\textbf{Interval} &\textbf{Degree} &\textbf{Depth}   \\\midrule
     Lee~\emph{et~al.}~\cite{lee2022privacy}&Remez.~\cite{lee2021high} &[-1,1] &50 &12\\   
     Lee~\emph{et~al.}~\cite{lee2022low}&Remez.~\cite{lee2021high} &[-1,1] &50 &12\\ 
     Kim~\emph{et~al.}~\cite{kim2023optimized}&Lee.~\cite{lee2022optimization} &[-1,1] &28 &10\\ 
     Rovida~\emph{et~al.}~\cite{rovida2024encrypted} &Chebyshev.~\cite{hernandez2001chebyshev} &[-1,1]&59 &5\\\hline
    \textbf{Ours}&Remez.~\cite{remez1934determination} &[-1,1] &\textbf{2}&\textbf{2}\\
     \bottomrule
  \end{tabular}
}
\end{table}
We have also evaluated the impact of our ReLU approximations on model accuracy, and the results are shown in Table~\ref{Tab:relu_accuracy}. By leveraging layer freezing and knowledge distillation techniques, we observed that the accuracy degradation remains within only $0.01\%$ and $0.02\%$ on MNIST and CIFAR-10 datasets, respectively. These results demonstrate that our ReLU approximation not only preserves model accuracy but also significantly enhances computational efficiency. 

\begin{table}[h]
  \caption{Comparison of model accuracy with ReLU and approximated ReLU activations on different datasets.}
  \label{Tab:relu_accuracy}
  \centering
  \resizebox{\linewidth}{!}{ 
  \begin{tabular}{l|l|cccccccc}
    \toprule
    \textbf{Dataset} &\textbf{Model}   &\textbf{Acc. (ReLU)} &\textbf{Acc. (Approx)}&\textbf{Loss$\downarrow$}   \\\midrule
     \multirow{2}{*}{MNIST}& $\textnormal{ResNet20}_{\textnormal{frozen-6}}$ &99.59\% &99.60\% &-0.01\%\\
    & $\textnormal{ResNet20}_{\textnormal{frozen-8}}$ &99.40\% &99.39\% &0.01\%\\\midrule
        \multirow{2}{*}{CIFAR-10}& $\textnormal{ResNet20}_{\textnormal{frozen-6}}$ &92.31\% &92.58\% &-0.27\%\\
    & $\textnormal{ResNet20}_{\textnormal{frozen-8}}$ &88.11\%&88.14\% &-0.03\%\\
     \bottomrule
  \end{tabular}
}
\end{table}


\subsection{Macro-benchmark Evaluations}\label{subsec:macrobenchmarks}

\textbf{End-to-End Performance.} 
We first present a summarization of the most recent works compared with our proposed method in Table~\ref{Tab:end_to_end_fhe}. It is obvious that $\textnormal{Pacifo}_{\{\textnormal{B}\}}$ achieves the best runtime performance, outperforms Kim~\emph{et~al.}~\cite{kim2023optimized} and Rovida~\emph{et~al.}~\cite{rovida2024encrypted} by a magnitude of $21.7\times$ and $25.2\times$, with only a limited accuracy degradation of $2\%$ compared with Kim~\emph{et~al.}~\cite{kim2023optimized}. Additionally, $\textnormal{Pacifo}_{\{\textnormal{L}\}}$ with more private parameters could achieve the highest accuracy of $92.58\%$ and is still $2.7\times$ and $3.1\times$ faster than Kim~\emph{et~al.}~\cite{kim2023optimized} and Rovida~\emph{et~al.}~\cite{rovida2024encrypted}, respectively. In the following, we will provide a more detailed analysis of Pacifo, including end-to-end performance, communication overhead, accuracy, etc.

\begin{table*}[h]
\centering
\resizebox{\linewidth}{!}{
\begin{threeparttable}
  \caption{Comparison of the inference performance for HE-based ResNet20 on CIFAR-10.}
  \label{Tab:end_to_end_fhe}
  \begin{tabular}{l|ccccccccc}
    \toprule
    \textbf{Framework} & \textbf{Scheme} & \textbf{\#Private Params} & \textbf{Processor} & \textbf{\#Threads} & \textbf{Runtime (s)} & \textbf{Accuracy} \\
    \hline
    Lee~\emph{et~al.}~\cite{lee2022privacy} & CKKS & 298.8K & Xeon Platinum 8280 CPU & 64 & 10602 & 90.67\%\\
    Lee~\emph{et~al.}~\cite{lee2022low} & CKKS & 278.3K & Threadripper PRO 3995WX@2.096GHz & 1 & 2271 & 91.31\%\\
    Benamira~\emph{et~al.}~\cite{benamira2023tt} & TFHE & 278.3K & i7-8650U@1.9GHz & 4 & 570 & 74.10\%\\
    Kim~\emph{et~al.}~\cite{kim2023optimized} & CKKS & 278.3K & EYPC 7402P@2.8GHz & 4 & 368 & 90.32\%\\
    Rovida~\emph{et~al.}~\cite{rovida2024encrypted} & CKKS & 278.3K & Xeon Gold 5218R CPU@2.1GHz & 8 & 426.5 & 91.67\%\\
    \hline
    $\textnormal{Pacifo}_{\{\textnormal{L}\}}$ (Ours) & CKKS & 206.3K & Xeon Gold 5218R CPU@2.1GHz & 8 & 135.1 & \textbf{\boldmath 92.58\%}$^{\dag}$\\
    $\textnormal{Pacifo}_{\{\textnormal{B}\}}$ (Ours) & CKKS & 74.6K & Xeon Gold 5218R CPU@2.1GHz & 8 & \textbf{16.9} & 88.19\%$^{\dag}$\\
    \bottomrule
  \end{tabular}
  \begin{tablenotes}
    \small
    \item[$\dag$] The results are evaluated on the full CIFAR-10 validation set (10,000 images), whereas some others use only subsets.
  \end{tablenotes}
\end{threeparttable}
}
\end{table*}

\textbf{Performance Breakdown.} 
Table~\ref{Tab:breakdown_fhe} provides a detailed runtime breakdown of each operation within a basic residual block of $\textnormal{Pacifo}_{\{\textnormal{L}\}}$. Notably, ``Conv+BatchNorm'' and ``bootstrapping'' are the most time-consuming components, requiring 9.39s and 7.97s, respectively, which correspond to $50.75\%$ and $43.07\%$ of the total runtime. Additionally, we observe that the optimized ReLU implementation contributes less than $0.05\%$ to the overall execution time, demonstrating its computational efficiency.


\begin{table}[ht]
  \caption{End-to-end runtime breakdown of a basic block in layer 3 of $\textnormal{Pacifo}_{\{\textnormal{L}\}}$.}
  \label{Tab:breakdown_fhe}
  \centering
  \resizebox{.98\linewidth}{!}{ 
  \begin{tabular}{cccccccccc}
    \toprule
    \textbf{Operation}  & \textbf{Available Depth}  &\textbf{Time (s)} &\textbf{\%} \\
    \midrule
    Conv+BatchNorm &2$\rightarrow$1  &4.60 &24.86\% \\
    Bootstrapping &1$\rightarrow$4  &4.11 & 22.21\% \\
    ReLU    &4$\rightarrow$2  &0.23 &0.01\% \\
    Conv+BatchNorm &2$\rightarrow$1  &4.79 &25.89\% \\
    Residual Connection &1$\rightarrow$1  &0.12 &0.01\% \\
    Bootstrapping &1$\rightarrow$4 &3.86 &20.86\% \\
    ReLU    &4$\rightarrow$2  &0.36 &0.02\% \\\hline
    Total   &-  &18.50 &- \\
    \bottomrule
  \end{tabular}
}
\end{table}

\textbf{Communication Overhead.} Here, we present the communication cost of Pacifo. Recall that we set the ring dimension into $N=2^{15}$, the number of available slots in one ciphertext will be $N'=N/2=16384$. Figure~\ref{fig:comm_cost} presents the dimension changes of an input image ($32 \times 32 \times 3$) as it passes through the ResNet20 model, and we can see that the maximum output size is $16384$ in layer 1. Hence, we are allowed to pack the output of an arbitrary block into a single ciphertext, significantly decreasing the communication cost. Since Pacifo is non-interactive, the client only needs to transmit one ciphertext, and the server returns the final inference result in one ciphertext. Specifically, in the case where $N=2^{15}$, the communication overhead for $\textnormal{Pacifo}_{\{\textnormal{L}\}}$ is roughly $11.5$MB, i.e., $10.5$MB (client)$+1.0$MB (server). For $\textnormal{Pacifo}_{\{\textnormal{B}\}}$, the overall communication cost is reduced to $10.5$MB, as the output size of Layer 3.1 is 4096, requiring fewer slots than in $\textnormal{Pacifo}{{\textnormal{L}}}$ (8192).

\begin{figure}[t]
\centering
\includegraphics[width=1.0\linewidth]{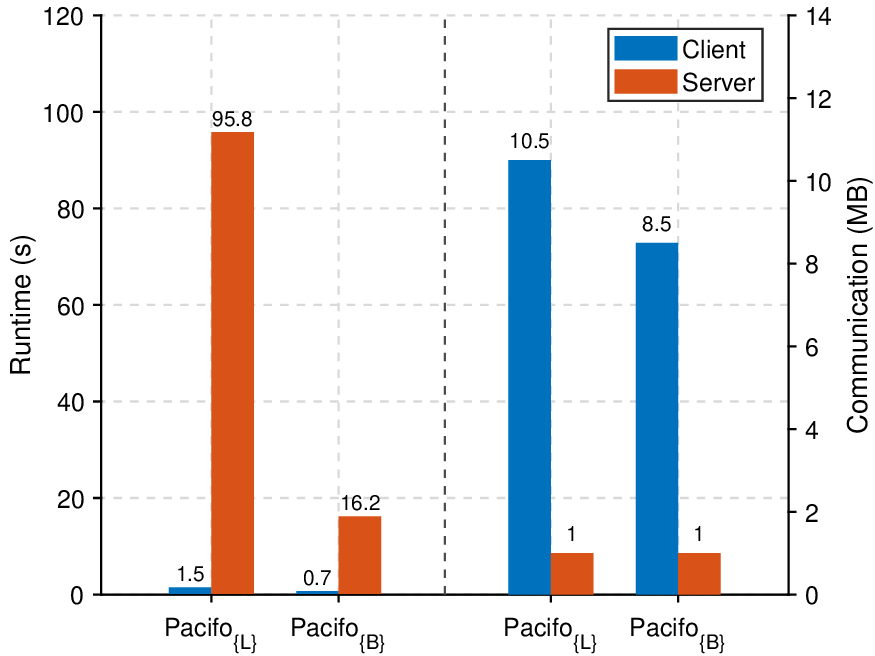}
    \caption{Flattened output size of different block in ResNet20 and the corresponding FHE ciphertext size (MB).} 
    \label{fig:comm_cost}
\end{figure}


\textbf{Accuracy and Precision.}
We evaluate the accuracy performance of Pacifo on three datasets: MNIST, CIFAR-10 and CIFAR-100. Table~\ref{Tab:accuracy_fhe} summarizes the inference accuracy of Pacifo when evaluated in the context of plaintext and ciphertext. Specifically, Pacifo in ciphertext demonstrates comparable accuracy performance to that of the plaintext. 

\begin{table}[h]
  \caption{Inference accuracy on different datasets.}
  \label{Tab:accuracy_fhe}
  \centering
  \resizebox{\linewidth}{!}{ 
  \begin{tabular}{l|lcccccccc}
    \toprule
    \textbf{System}  &\textbf{Dataset} &\textbf{Plaintext} &\textbf{Ciphertext} &\textbf{Loss$\downarrow$} \\\midrule
    \multirow{3}{*}{$\textnormal{Pacifo}_{\{\textnormal{L}\}}$} 
     &MNIST &99.65\%&99.63\% &0.02\%\\
     &CIFAR10 &92.58\% &92.58\% &0.00\%\\\midrule
    \multirow{3}{*}{$\textnormal{Pacifo}_{\{\textnormal{B}\}}$} 
     &MNIST &99.39\% &99.40\% &-0.01\% \\
     &CIFAR10 &88.14\% &88.19\% &-0.05\% \\
     \bottomrule
  \end{tabular}
}
\end{table}

In addition, to further investigate the influence of FHE noise on model accuracy, we follow Rovida~\emph{et~al.}~\cite{rovida2024encrypted} and introduce the concept of relative precision, quantifying the accuracy of an approximation relative to the correct values. We compute the value of relative precision ($P$) as follows:
\begin{align} 
P = 1 - \frac{\frac{1}{n} \sum_{i=1}^{n} |c_i - a_i|}{\frac{1}{n} \sum_{i=1}^{n} |c_i|} \end{align}
where $c_i$ and $a_i$ represent the correct and approximate values, respectively, and $n$ is the total number of values. A higher precision value indicates a closer approximation to the ground truth. Fig.~\ref{fig:precision} presents the evaluated relative precision results throughout the inference process. As noise accumulates, a gradual decline in $P$ is observed. However, $P$ always remains above $0.98$, ensuring that the ciphertext model maintains precision comparable to the plaintext one. This also illustrates why the accuracy degradation in Table~\ref{Tab:accuracy_fhe} is limited.

\begin{figure}[t]
\centering
\includegraphics[width=0.9\linewidth]{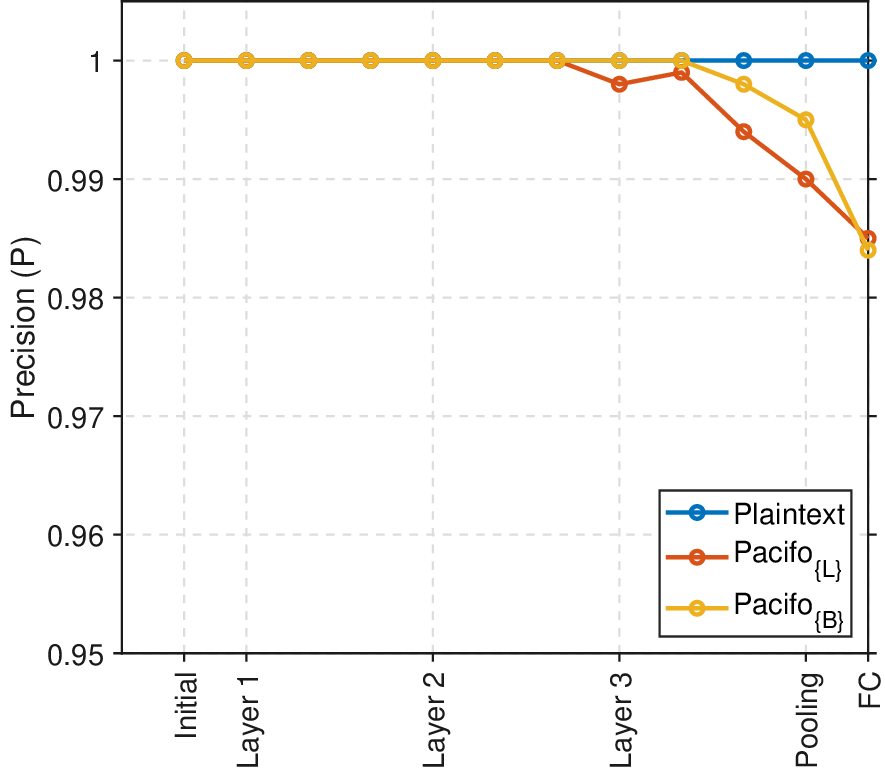}
    \caption{Comparision of relative precision of Pacifo with plaintext model.} 
    \label{fig:precision}
\end{figure}


\subsection{Scalability on Tuning HE Parameters.} 

In this subsection, we explore how to adjust FHE parameters to balance security and efficiency. While previous experiments maintain a $128$-bit security level for fair comparisons with existing works, we further analyze how different security levels can be tailored to specific applications.

Highly sensitive information, such as financial transactions or government communications, demands higher security (e.g., $128$-bit or above) to defend against computational attacks, including those posed by quantum threats in future. Based on our discussion in Section~\ref{Sec:security_standards}, we adopt two additional security levels, i.e., $192$-bit and $256$-bit, and compare the end-to-end runtime performance of $\textnormal{Pacifo}_{\{\textnormal{L}\}}$ in Table~\ref{Tab:Effe_performance}.

\begin{table}[h]
  \caption{Runtime performance of $\textnormal{Pacifo}_{\{\textnormal{L}\}}$ under different security levels.}
  \label{Tab:Effe_performance}
  \centering
  \resizebox{\linewidth}{!}{ 
  \begin{tabular}{c|cc|c}
    \toprule
    \textbf{Security level} &  \textbf{Ring size} & \textbf{Moduli chain size} & \textbf{Runtime (s)}  \\\midrule
     $128$-bit & $2^{16}$ &$1228$ bits &$135.1$  \\
     $192$-bit & $2^{16}$ &$1188$ bits &$149.7$  \\
     $256$-bit & $2^{17}$ &$1289$ bits &$256.5$ \\
     \bottomrule
  \end{tabular}
}
\end{table}

\textbf{192-bit Security.}
To achieve a $192$-bit security, we could matain the ring dimension $N$ as $2^{16}$, but slightly decrease the moduli chain size into below $1200$ bits. The change of moduli chain size will not significantly affect the computation complexity, and thus we can obtain an end-to-end runtime of $149.7$ s, only slightly higher than the time of $135.1$ s under standard $128$-bit security.

\textbf{256-bit Security.} To achieve a $256$-bit security, we have to raise $N$ into $2^{17}$ to strengthen the lattice hardness assumptions, and the moduli chain size is also expanded correspondingly. The change of $N$ substantially raise the overall computational overhead. As a result, the end-to-end runtime increases to $256.5$~s, representing nearly a $1.9\times$ slowdown compared to the $128$-bit configuration.


\section{Conclusion}\label{Sec:conclusion}
In this paper, we proposed Pacifo, an efficient FHE-based private inference framework for MLaaS. By combining layer partitioning, lightweight ReLU approximation, and flexible security configurations, Pacifo significantly reduces computational overhead while preserving model accuracy and privacy. Experimental results show up to $21.7\times$ speedup over existing solutions with minimal accuracy loss, along with strong scalability to higher security levels. These results demonstrate that Pacifo makes privacy-preserving deep learning inference more practical and adaptable for real-world applications.


%
\bibliographystyle{IEEEtran}
\bibliography{ref}

\appendices




\ifCLASSOPTIONcaptionsoff
  \newpage
\fi

\end{document}